\documentclass[a4]{llncs}
\usepackage[T1]{fontenc}

\usepackage[font=small]{caption}
\usepackage{latexsym}
\usepackage{graphicx}
\usepackage{tikz}
\usepackage{color}
\usepackage{url}
\usepackage{subfig}
\usepackage{amsmath,amssymb}
\usepackage{verbatim}

\begin{document}


\title{Counting approximately-shortest paths \\ in directed acyclic graphs}
\author{
  Mat\'u\v{s} Mihal\'ak \and
  Rastislav \v{S}r\'amek
  \and Peter Widmayer}
\institute{Institute of Theoretical Computer Science, ETH Zurich, Zurich,
  Switzerland \\
  e-mail:~\{mmihalak,rsramek,widmayer\}@inf.ethz.ch}

\maketitle
\begin{abstract}
  Given a directed acyclic graph with positive edge-weights, two vertices $s$
  and $t$, and a threshold-weight $L$, we present a fully-polynomial time
  approximation-scheme for the problem of counting the $s$-$t$ paths of length
  at most $L$.
  We extend the algorithm for the case of two (or more) instances of the same
  problem. That is, given two graphs that have the same vertices and edges and
  differ only in edge-weights, and given two threshold-weights $L_1$ and $L_2$,
  we show how to approximately count the $s$-$t$ paths that have length at most
  $L_1$ in the first graph and length not much larger than $L_2$ in the second
  graph.
  We believe that our algorithms should find application in counting approximate
  solutions of related optimization problems, where finding an (optimum)
  solution can be reduced to the computation of a shortest path in a
  purpose-built auxiliary graph.
  %
\end{abstract}

\section{Introduction \label{chap:introduction}}

Systematic generation and enumeration of combinatorial objects (such as graphs,
set systems, and many more) has been a topic of extensive study in the field of
combinatorial algorithms for decades \cite{kreher98}.
Counting of combinatorial objects has been investigated at least as thoroughly,
even leading to their own computational complexity class \#P, defined in
Valiant's seminal paper~\cite{valiant79}.
A counting problem usually asks for the number of solutions to a given
combinatorial problem, such as the number of perfect matchings in a bipartite
graph. In combinatorial optimization, the number of optimum solutions can
sometimes be computed by a modification of an algorithm for finding a single
optimum solution. For instance, for shortest $s$-$t$ paths in graphs with
positive edge weights, Dijkstra's algorithm easily admits such a modification.
The problem we discuss in this paper has a more general flavor: We aim at
counting the number of approximate solutions, in the sense of solutions whose
objective value is within a given threshold from optimum. For shortest
$s$-$t$ paths, it is not obvious how to count the number of paths within, say,
10\% from optimum. 
%
A related problem of enumerating feasible solutions makes a step in this
direction: If we can enumerate solutions in order of decreasing quality,
starting from an optimum solution, we have a way to count approximate
solutions. Even though for some problems there are known enumeration
algorithms that return the next feasible solution in the sequence of solutions
within only polynomial extra time (called ``polynomial delay''), this approach
will usually not be satisfactory in our setting. The reason is that the number
of approximate solutions can be exponential, and counting by enumerating then
takes exponential time, while our interest is only in the count itself.
    
In this paper we propose a way to count approximate solutions for the shortest
$s$-$t$ path problem in directed acyclic graphs (DAGs) in polynomial time, but
the count that we get is only approximate, even though we come as close to the
exact count as we wish (technically, we propose an FPTAS). We also show that
exact counting for our problem is \#P-hard, thus (together with the FPTAS) fully
settling its complexity. 
We achieve our result by a modification of a conceptually interesting dynamic
program for all feasible solutions for the knapsack problem
\cite{stefankovic2012}. Our motivation for studying our counting problem comes
from a new approach \cite{buhmann2013} to cope with uncertainty in optimization
problems. There, we not only need to count the number of approximate solutions
for a given problem instance, but we also need to count the number of solutions
that are approximate (within a given approximation ratio) for two problem
instances at the same time. For the case of shortest $s$-$t$ paths, this means
that we are given two input graphs that are structurally identical, but are
allowed to differ in their edge weights. We now want to count the number of
$s$-$t$ paths that are within, say, 10\% from optimum in both input graphs at
the same time. For this problem we propose both a pseudo-polynomial algorithm
and an algorithm that calculates an approximate solution for a potentially
slightly different threshold in fully polynomial time. 
Our hope is that our study paves the way for approximately counting approximate
solutions for other optimization problems, such as minimum spanning trees.

The rest of the paper is organized as follows. We outline possible implications
of our result in Section~\ref{sec:applications_of_shortest_paths_in_DAGs}. We
show in Section~\ref{sec:NP-hardness} that our problem is {\#}P-complete. We
present the algorithms in Section~\ref{sec:fptas}, and conclude the paper in
Section~\ref{sec:conclusions}.

\subsection{Dynamic Programming as Shortest-Path Computation in DAGs}
\label{sec:applications_of_shortest_paths_in_DAGs}

The concept of computing a shortest $s$-$t$ path in a directed acyclic graph has
a large number of applications in many areas of algorithmics.
%
%
This is partly due to the fact that dynamic programming algorithms in which the 
inductive step consists of searching for a maximum or a minimum among some functions 
of previously-computed values can be viewed as the problem of looking for the shortest 
or longest path in a
directed acyclic graph.\footnote{Note that due to the lack of cycles, the
problems of looking for shortest and longest paths on DAGs are computationally
identical.} 

In many problems that admit a dynamic programming solution we are interested not
only in the single optimum, but also in other approximately optimal solutions.
For instance, if we single out the context of analysis of biological data, de
novo peptide sequencing \cite{chen2001,lu03}, sequence alignment
\cite{naor1993}, or Viterbi decoding of HMMs \cite{burge1997,durbin1998} all use
dynamic programming to find a shortest path in some implicit graph.
Due to the nature of the data in these applications, producing a single solution
is often insufficient and enumerating all solutions close to the optimum is
necessary. Our contribution, therefore, provides a faster solution than explicit 
enumeration for the problems where counting of approximate solutions is required \cite{naor1993}.
Counting and sampling from close-to-optimum solutions is the key-element of the
recent optimization method with uncertain input data of Buhmann et
al.\,\cite{buhmann2013}. Our work thus makes a step towards practical algorithms
in this context. 


\subsection{Counting Approximate Solutions is \#P-Complete}
\label{sec:NP-hardness}

The problem of counting the number of all self-avoiding $s$-$t$ walks in a
directed (or undirected) graph is known to be \#P-complete \cite{Valiant/1979}.
The proof makes use of graphs containing cycles, thus it cannot be used to show
the hardness of the problem of counting approximate shortest paths on a directed
\emph{acyclic} graph. In fact, we can easily count all $s$-$t$ paths in a
directed acyclic graph in time proportional to the number of edges, if we
traverse the graph vertices sorted in topological order and add up the number of
paths arriving to each vertex from its predecessors. The difficulty thus lies in
the addition of edge-weights and the requirement to count $s$-$t$ paths of
length at most $L$.
In the following, we show that this problem is \#P-complete, by
a reduction from the $NP$-complete \emph{partition problem}.
Given a set of positive integers $S=\{s_1,\ldots,s_n\}$,
the partition problem asks for a partition of $S$ into sets
$S_1$ and $S_2$ such that the sums of numbers in both sets are equal.

Given an instance $S=\{s_1,\ldots,s_n\}$ of the partition problem, we construct
a graph with $n+1$ vertices $v_1,\ldots,v_{n+1}$ as follows. We consider the
elements of $S$ in an arbitrary order $s_1,\ldots,s_n$. Then, for every $i<n$,
the graph will contain two parallel edges between vertices $v_i$ and $v_{i+1}$
with lengths $s_i$ and $-s_i$, respectively. Then every path from $v_1$ to
$v_{n+1}$ corresponds to one partition of $S$ to subsets $S_1$ and $S_2$. If,
between two consecutive vertices $v_i$ and $v_{i+1}$, the edge with length $s_i$
is chosen, $s_i$ will belong to the set $S_1$. If the chosen edge has length
$-s_i$, the element $s_i$ will belong to the set $S_2$. The length of the
$v_1$-$v_{n+1}$ path then corresponds to the difference between the sums of
elements in $S_1$ and in $S_2$ and
the number of paths of length $0$ is then equal to the number of optimal solutions 
of the partition problem. 

If we had an algorithm that can count the number of $v_1$-$v_{n+1}$ paths of length 
at most $-1$ and the number of $v_1$-$v_{n+1}$ paths of length at most $0$, 
the difference between these two numbers is the number of paths of length 
exactly $0$ and thus the number of solutions to the partition problem.
%
%

Since the partition problem is reducible from the \#P-complete 
knapsack problem \cite{dyer1993mildly} and 
its own reduction as well as ours is parsimonious \cite{karp1972}, the problem of counting all 
$s$-$t$ paths of length at most $L$ is \#P-complete.
Note that the existence of parallel edges is not necessary for the reduction; we
could bisect each parallel edge creating an auxiliary vertex to form a graph of
the same functionality but without parallel edges. Also, observe that the use of
negative edge-weights is not necessary; we can add to every edge-weight a very
large number $M$ (say, the maximum number in $S$), and then ask whether there
exists a path of length $nM$.
Thus, we have shown the following.

\begin{theorem}
  Let $G$ be a directed acyclic graph with integer edge-weights, and $L$ be
  an integer. The problem of counting all $s$-$t$ paths of length at most $L$ is
  \#P-complete, even if all edge-weights are non-negative.
\end{theorem}




\section{Approximation Algorithms \label{sec:fptas}}

In this section we present an FPTAS for our counting problem. That is, we
present an algorithm that when given a directed acyclic graph $G$ on $n$
vertices, two dedicated vertices $s$ and $t$, a weight-threshold $L$, and a
constant $\varepsilon>0$, computes a $(1+\varepsilon)$-approximation of the
total number of $s$-$t$ paths of length at most $L$, and which runs in time
polynomial in both $n$ and $\frac{1}{\varepsilon}$.

Let us note why the most immediate attempt to solve the problem directly 
does not work. We could try to calculate the number of paths from $s$ to each 
vertex $i$ that are shorter than all possible thresholds $L$. We can do this 
incrementally by calculating the paths for vertices sorted in topological order 
and for each new vertex combining the paths that arrived from previously computed 
vertices. We can then pick some polynomially large subset of the thresholds $L$ 
and round all distances down to the nearest one in the subset. 
While we would end up with an algorithm of polynomial run-time, 
it would not constitute a FPTAS, since we would exactly count the number of paths 
that are no longer than some length $L'$ which does not differ much from 
our desired maximum length $L$, instead of approximately counting the 
number of solutions that are shorter than the exact length $L$. 

We first show a recurrence that can be used to exactly count the number of
$s$-$t$ paths of length at most $L$.
Evaluating the recurrence takes exponential time, but we will later show how to
group partial solutions together in such way that we trade accuracy for the
number of recursive calls. We adapt the approach of \v{S}tefankovi\v{c} et
al.\,\cite{stefankovic2012}, which they used to approximate the number of all
feasible solutions to the knapsack problem. 

Let $G$ be a directed acyclic graph with $n$ vertices. We will label the
vertices $v_1,\ldots, v_n$ in such order that there is no path from $v_i$ to
$v_j$ unless $i < j$, i.e., $v_1,\ldots, v_n$ defines a topological ordering. We
suppose that $v_1=s$ and $v_n=t$, otherwise the graph can be pruned by
discarding all vertices that appear before $s$ and after $t$ in the topological
order, since no path from $s$ to $t$ ever visits these. 

Now, for a given $L$, instead of asking for the number of $s$-$t$ paths that
have length at most $L$, we indirectly ask: for a given value $a$, what is the
smallest threshold $L'$ such that there are at least $a$ paths from $s$ to $t$
of length at most $L'$? Let $\tau(v_i,a)$ denote the minimum length $L'$ such
that there are at least $a$ paths from $v_1$ to $v_i$ of length at most $L'$.  
To find the number of $s$-$t$ paths of length at most $L$ using this function
$\tau$, we simply search for the largest $a$ such that $\tau(v_n,a) \leq L$,
and return it as the output. In particular, if the length of the shortest
$s$-$t$ path is $OPT$ (which can be computed in polynomial time), we can find,
for any $\rho>1$, the number of $\rho$-approximate $s$-$t$ paths by setting
$L:=\rho OPT$.

For a concrete vertex $v_i$ with in-degree $d_i$, let us denote its $d_i$
neighbors that precede it in the topological order by
$p_{1},\ldots,{p_{d_i}}$ and let us denote the corresponding incoming
edge lengths by $l_{1},\ldots,l_{d_i}$. For simplicity, we usually drop the
index $i$ when it is clear from the context and just write $d$, $p_1,\ldots,p_d$
and $l_1,\ldots,l_d$. 
Now, $\tau(v_i,a)$ can be expressed by the following recurrence 

\begin{eqnarray*}
  &&\tau(v_1,0)=-\infty\\
  &&\tau(v_1,a)=0, \forall a: 0<a\leq 1 \\
  &&\tau(v_1,a)=\infty, \forall a: a > 1 \\
  &&\tau(v_i,a)=\min_{\substack{\alpha_1,\ldots,\alpha_d\\\sum\alpha_j=1}}
  \max_s (\tau(p_s,\alpha_s a)+l_s) \text{.}
  \label{eq:tau}
\end{eqnarray*}

Intuitively, the $a$ paths starting at $v_1$ and arriving at $v_i$ must
split in some way among incoming edges. The values $\alpha_j$ define such split.
We look for a set of $\alpha_1,\ldots,\alpha_{d}$ that minimizes the maximum
allowed path length needed such that the incoming paths can be distributed
according to $\alpha_j$, $j=1,\ldots,d$. Note that while the values of $\alpha_i
a$ do not have to be integer, $\tau(v_i,\alpha_i a)$ is equal to
$\tau(v_i,\lceil\alpha_i a\rceil)$. Moreover, when evaluating the recursion, it
is enough to search for values $\alpha_i$ such that each of the values $\alpha_1
a,\ldots,\alpha_d a$ is an integer.


Calculating $\tau$ using the given recurrence will not result in a polynomial
time algorithm since we might need to consider an exponential number of values
for $a$, namely $2^{n-2}$ on a DAG with a maximal number of edges.%
\footnote{To see this, observe that in a topologically sorted graph $G$, any
  subset of $V\setminus \{s,t\}$ gives a unique candidate for an $s$-$t$ path.}
To overcome this, we will consider only a polynomial number of possible values
for $a$, and always round down to the closest previously considered one in the
recursive evaluation. If we are looking for an algorithm that counts with
$1+\varepsilon$ precision, the ratio between two successive considered values of
$a$ must be at most $1+\varepsilon$. 

For this purpose, we introduce a new function $\tau'$. In order to achieve
precision of $1+\varepsilon$, we will only consider values of $\tau'$ for
minimum path numbers in the form of $q^k$ for all positive integers $k$ such
that $q^k < 2^{n-2}$, where $q=\sqrt[n+1]{1+\varepsilon}$.  The values of
$\tau'$ for other numbers of paths will be undefined. The function $\tau'$ is
defined by the recurrence

\begin{eqnarray}
  &&\tau'(v_1,0)=-\infty \nonumber\\
  &&\tau'(v_1,a)=0, \forall a: 0<a\leq 1 \nonumber\\
  &&\tau'(v_1,a)=\infty, \forall a: a > 1 \nonumber\\
  &&\tau'(v_i,q^j)=\min_{\substack{\alpha_1,\ldots,\alpha_d\\\sum\alpha_j=1}}
  \max_s (\tau'(p_s, q^{\lfloor j+\log_q\alpha_s\rfloor})+l_s)\text{.}
  \label{eq:tauprime}
\end{eqnarray}

To give a meaning to the expression $q^{\lfloor j+\log_q\alpha_i\rfloor}$ when
$\alpha_i=0$, we define it to be equal to $0$, which is consistent with its
limit when $\alpha_i$ goes to $0$. We now show that the rounding does not
make the values of $\tau'$ too different from the values of $\tau$. 
%

\begin{lemma}
  Let $1\leq i$ and $i\leq j$. Then
  \begin{equation}
    \tau(v_i,q^{j-i}) \leq \tau'(v_i,q^j) \leq \tau(v_i,q^j)\textrm{.}
    \label{eq:1}
  \end{equation}
  \label{lemma:1}
\end{lemma}
\begin{proof}
We first prove the first inequality, proceeding by induction on $i$. The base
case holds since $\tau(v_1,a)\leq\tau'(v_1,b)$ for any $a\leq b$. Suppose now
that the first inequality of (\ref{eq:1}) holds for every $p$, $p<i$. Then, for
every $0 \leq \alpha < 1$,
\begin{eqnarray*}
  & & \tau'(p,q^{\lfloor j+\log_q\alpha\rfloor})\geq 
  \tau(p,q^{\lfloor j+\log_q\alpha\rfloor - p}) \\
  &\geq & \tau(p, q^{j-p-1+\log_q\alpha}) \geq \tau(p, \alpha q^{j-i})
  \textrm{.}
\end{eqnarray*}

Thus, since every predecessor of $v_i$ is earlier in the vertex ordering, we can
use the obtained inequality to get the claimed bound

\begin{eqnarray*}
  \tau'(v_i,q^j) & = & \min_{\substack{\alpha_1,\ldots,\alpha_d\\\sum\alpha_j=1}}
\max_s \tau'(p_s, q^{\lfloor j+\log_q\alpha_s\rfloor})+l_s\\
  & \geq &
  \min_{\substack{\alpha_1,\ldots,\alpha_d\\\sum\alpha_j=1}}
\max_s\tau(p_s, \alpha_s q^{j-i})+l_s
   =  \tau(v_i,q^{j-i}) \textrm{.}
\end{eqnarray*}

The other inequality $\tau'(v_i,q^j)\leq\tau(v_i,q^j)$ follows by a simpler
induction on $i$. The base case holds since $\tau(v_1,x)=\tau'(v_1,x)$ for all
$x$. Assume now that the second part of (\ref{eq:1}) holds for all $p<i$. Then
\begin{equation*}
  \tau'(p,q^{\lfloor j+\log_q\alpha_i\rfloor}) \leq 
  \tau(p,q^{\lfloor j+\log_q\alpha_i\rfloor}) \leq 
  \tau(p,\alpha_i q^{j})\textrm{.}
\end{equation*}

We can now use the recursive definition to obtain the claimed inequality
$\tau'(v_i,q^j)\leq\tau(v_i,q^j)$:

\begin{eqnarray*}
  \tau'(v_i,q^j) & = & \min_{\substack{\alpha_1,\ldots,\alpha_d\\\sum\alpha_j=1}}
\max_s \tau'(p_s, q^{\lfloor j+\log_q\alpha_s\rfloor})+l_s \\
  & \leq &
  \min_{\substack{\alpha_1,\ldots,\alpha_d\\\sum\alpha_j=1}}
\max_s \tau(p_s, \alpha_s q^{j})+l_s
   =  \tau(v_i,q^{j}) \textrm{.}
\end{eqnarray*}
\qed
\end{proof}

We can now use $\tau'(v_n,q^k)$ to obtain a $(1+\varepsilon)$-approximation for
the counting problem. Basically, for any $L$, we show that for the largest 
integer $k$ such that $\tau'(v_n,q^k)\leq L < \tau'(v_n,q^{k+1})$, the value $q^k$
will be no more than $(1+\varepsilon)^{\pm 1}$ away from the optimum.

\begin{lemma}
  Given $L$, let $k$ be such that $\tau'(v_n,q^k)\leq L < \tau'(v_n,q^{k+1})$
  and $a$ be such that $\tau(v_n,a)\leq L < \tau(v_n,a+1)$. Then
  $(1+\varepsilon)^{-1}\leq \frac{a}{q^k} \leq 1+\varepsilon$.
\label{thm:1}
\end{lemma}
\begin{proof}
Using Lemma \ref{lemma:1} twice, we get $\tau(v_n,q^{k-n}) \leq
\tau'(v_n,q^k) \leq L < \tau'(v_n,q^{k+1})\leq\tau(v_n,q^{k+1})$.
As $\tau(v_n,q^{k-n})$ is at most $L$, and $a$ is largest such that $\tau(v_n,a)
\leq L$, and $\tau$ is monotonous in its second parameter, it must be that
$q^{k-n}\leq a$. 
Similarly, $\tau(v_n,q^{k+1})$ is larger than $L$, so by monotonicity $a\leq
q^{k+1}$. Thus both $a$ and $q^k$ must lie between $q^{k-n}$ and $q^{k+1}$ and
their ratio can be at most $q^{k+1-(k-n)} = q^{n+1} = 1+\varepsilon$ and at
least $q^{k-(k+1)} = (1+\varepsilon)^{-1/(n+1)} > (1+\varepsilon)^{-1}$.
\qed
\end{proof}

We now show that computing the values of $\tau'(v_i,q^k)$ can be done in time
polynomial in $n$ and $\frac{1}{\varepsilon}$. This then, together with
Lemma~\ref{thm:1}, gives an FPTAS for the counting problem.

\begin{theorem}
  For any $L$, any edge-weighted directed acyclic graph $G$, and any vertices
  $s$, $t$, there is an FPTAS that counts the number of all
  $s$-$t$ paths in $G$ of length at most $L$ in time $O(mn^3\varepsilon^{-1}\log{n})$.
%
  \label{thm:3}
\end{theorem}

\begin{proof}
Recall that a directed acyclic graph on $n$ vertices has at most $2^{n-2}$
$s$-$t$ paths.
The values of $a$ in $\tau$ therefore span at most $\{1,2,\ldots,2^{n-2}\}$, and
the values of $q^k$ in $\tau'$ span at most $\{1,q,q^2,\ldots,q^s\}$, where $$s
:= \log_q(2^{n-2}) = {(n-2) \over \log_2{q}} = {(n-2)(n+1) \over
\log_2(1+\varepsilon)} = O(n^2\varepsilon^{-1}).$$ Thus, we evaluate function
$\tau'$ for at most $ns = O(n^3\varepsilon^{-1})$ different parameter pairs. 

To show that the evaluation of $\tau'$ can be done in polynomial time, we need
to show that we can efficiently find $\alpha_1,\ldots,\alpha_d$ that minimize
Expression (\ref{eq:tauprime}). Fortunately, $\tau'(v_i,q^k)$ is monotonous with
increasing $k$, we can thus apply a greedy approach. Given $v_i$, we will
evaluate $\tau'(v_i,q^k)$ for all possible values of $q^k$ in one run. Instead
of looking for the tuple $\alpha_1,\ldots,\alpha_d$ such that $\sum\alpha_i=1$ we will consider an integer tuple
$k_1,\ldots, k_d$ that minimizes $\max_i\tau'(p_i,q^{k_i})$ restricted by $\sum q^{k_i}>q^{k-1}$. 
We start with all $k_i$ equal to $0$ and always increase by
one the $k_i$ that minimizes $\tau'(p_i,q^{k_i+1})+l_i$. 
Whenever the sum of all $q^{k_i}$ gets larger than some value $q^{k-1}$, we store
the current maximum of $\tau'(p_i,q^{k_i}) + l_i$ as the value $\tau'(v_i,q^{k})$. 
We terminate once $\sum_i{q^{k_i}}$ reaches $2^{n-2}$. It can be shown that such approach 
calculates the same values of $\tau'$ as searching through ratios $\alpha_i$. 
As we can increase each $k_i$ at most $s$ times, we make at most $ds$ steps,
each of which involves choosing a minimum from $d$ values and replacing it with
a new value. The latter can be done in time $O(\log{d})\subseteq O(\log{n})$,
for instance by keeping the values $\tau'(v_i,q^{k_i+1})+l_i$ in a heap. The sum
of the $d$'s for all considered vertices is equal to the number of edges $m$.
The update of $\sum_i{q^{k_i}}$, calculation of $q^{k+1}$ from $q^k$, and
comparison with the maximum number of paths can all be done in $O(\log(2^n)) =
O(n)$ time if we choose $q$ in the form $1+2^{-t}$ in order to be able to
implement multiplication by $q$ by a sequence of bit-shifts and a single
addition.
The resulting bit-time complexity is thus $O(mn^3\varepsilon^{-1}\log{n})$.
\qed
\end{proof}

We note that processing the dynamic programming table for all path numbers in
one go would to improve the time complexity of the original Knapsack FPTAS
\cite{stefankovic2012} by a factor of $O(\log(n))$.

\subsection{Counting solutions of given lengths in multiple instances \label{sec:bicriteria}}

In this section we consider the problem of counting solutions that are
approximately-optimum for two given instances at the same time. The two
instances differ in edge lengths, but share the same topology,
%
effectively forming a bi-criteria instance.
Formally, given two directed acyclic graphs $G_1$ and $G_2$, differing only in
edge-weights, given two vertices $s$ and $t$, and given two threshold values
$L_1$ and $L_2$, we are interested in the number of the $s$-$t$ paths that have
at the same time length at most $L_1$ in $G_1$ and length at most $L_2$ in
$G_2$.

To solve this algorithmic problem, we cannot directly apply the approach for the
single-instance case (by defining $\tau$ to be a pair of path lengths, one for
each of the two instances), as we now have two lengths per edge and it is
unclear how to suitably define a maximum over pairs in Equation (\ref{eq:tau}).
In fact, we can show that we cannot construct a FPTAS for the two instance 
scenario, or indeed any approximation algorithm.

\begin{theorem} Let $G_1$ and $G_2$ be two directed acyclic graphs with the same 
sets of vertices and edges, but possibly different edge-weights, let $s$ and 
$t$ be two vertices in them, let $L_1$ and $L_2$ be two length thresholds. 
The existence of an algorithm that in time polynomial in number of vertices 
$n$ computes any finite approximation of the 
number of paths from $s$ to $t$ that are shorter than $L_1$ if measured in 
the graph $G_1$ and shorter than $L_2$ if measured in the graph $G_2$, implies 
that $P=NP$.
\end{theorem}

\begin{proof}
We show this by reducing the decision version of the knapsack problem to 
the aforementioned problem. Let us have a knapsack instance with $n$ items 
with weights $w_1,\ldots,w_n$ and prices $p_1,\ldots,p_n$. Given a 
total weight limit $W$ and a price limit $P$ we want to know if we can select 
a set of items such that the total weight is at most $W$ and the 
total price is at least $P$. The corresponding DAG will have $n+1$ 
vertices $v_0,\ldots,v_n$, with two edges between all successive 
vertices $v_k$ and $v_{k+1}$ that will correspond to the action of 
taking or not taking the $k+1$-st element into the knapsack.
The first edge between $v_k$ and $v_{k+1}$ will have length 
$w_{k+1}$ in the graph $G_1$ and length $\frac{2P}{n+1}-p_{k+1}$ in the graph $G_2$, 
the second edge will have length $0$ in the graph $G_1$ and $\frac{2P}{n+1}$ in the 
graph $G_2$. We can now ask for the number of paths from $v_0$ to $v_n$ that are 
shorter than $W$ in the graph $G_1$ and shorter than $P$ in the graph $G_2$. If we 
had an algorithm that gives us a number that differs from this number by any real and finite 
multiplicative ratio $c$, we could determine whether the original knapsack 
problem had at least one solution since the ratio between $1$ and $0$ is not 
a real number. 
\qed
\end{proof}

This proof is perhaps surprising due to the fact that Gopalan et al. \cite{gopalan2011}
showed a FPTAS that counts the number of solutions of multi-criteria knapsack instances. This 
shows that while knapsack is a special version of our problem, it is in fact less 
complex due to the common assumption that the item values are non-negative. 

While we cannot obtain a $(1+\varepsilon)$-approximation of the number of
$s$-$t$ paths that have length at most $L_1$ in the first instance, and at the
same time length at most $L_2$ in the second instance, we will adopt the
techniques for FPTAS in a single instance, and show a polynomial-time algorithm
that provides heuristics for good estimates of $s$-$t$ paths that have length at
most $(1+\delta)L_1$ in the first instance, and at the same time length at most
$L_2$ in the second instance. We will only consider the case where $L_1$ is positive. 

To do so, we define a function $\tau_2$ similar in spirit to $\tau$ that uses
a maximum path-length $L_1$ in the form of a ``budget'' as a parameter of
$\tau_2$.
%
%
Formally, $\tau_2(v_i, a, L_1)$ is the smallest length $L_2$
such that there are at least $a$ $v_1$-$v_i$ paths, each of length at most $L_1$
with respect to the edge lengths in the first instance, and of length at most
$L_2$ with respect to the edge length in the second instance. 
Similarly to $\tau$, we can express $\tau_2$ recursively using the following
notation. Let $v_i$ be a vertex of in-degree $d$, and let $p_1,\ldots,p_{d}$
be the neighbors of $v_i$ preceding it in the topological order. The edge-length
of the incoming edge $(p_j,v_i)$, $j=1,\ldots,d_i$, is $l_j$ in the first
instance, and $l_j'$ in the second instance. 
Then, $\tau_2$ satisfies the following recursion:

\vskip-3ex
%
\begin{eqnarray*}
  &&\tau_2(v_1,0,x)=-\infty, \forall x\in \mathbb{R^+}\\
  &&\tau_2(v_1,a,x)=0, \forall a: 0<a\leq 1, \forall x\in \mathbb{R^+}\\
  &&\tau_2(v_1,a,x)=\infty, \forall a: a > 1, \forall x\in \mathbb{R^+}\\
  &&\tau_2(v_i,a,L_1)=\min_{\substack{\alpha_1,\ldots,\alpha_d\\\sum\alpha_j=1}}
\max_s \tau_2(p_s,\alpha_s a, L_1-l_{s})+l'_{s}
\end{eqnarray*}

If we wanted to use $\tau_2$ to directly use to solve our counting problem, the
function $\tau_2$ would have to be evaluated not only for an exponential number
of path counts $a$, but also for possibly exponential number of values of $L_1$.
To end up with polynomial runtime, we thus need to consider only a polynomial
number of values for both parameters of $\tau_2$.
%
For this purpose, we will introduce a function $\tau'_2$ that does this by
considering only path lengths in the form of $r^k$, where $r =
\sqrt[n]{1+\delta}$, and path numbers $a$ in the form of $q^j$, where $q =
\sqrt[n]{1+\varepsilon}$, for positive $\varepsilon$ and $\delta$.
Function $\tau'_2$ is defined by the following recurrence:
\begin{eqnarray*}
  &&\tau'_2(v_1,0,x)=-\infty, \forall x\in \mathbb{R^+} \\
  &&\tau'_2(v_1,a,x)=0, \forall a: 0<a\leq 1, \forall x\in\mathbb{R^+} \\
  &&\tau'_2(v_1,a,x)=\infty, \forall a: a > 1, \forall x\in \mathbb{R^+} \\
  &&\tau'_2(v_i,q^j,r^k)=\min_{\substack{\alpha_1,\ldots,\alpha_d\\\sum\alpha_j=1}}
\max_s \tau'_2(p_s,q^{\lfloor j+\log_q\alpha_s\rfloor}, r^{\lfloor \log_r(r^k-l_s)\rfloor})+l'_s
\end{eqnarray*}

Similarly to the case of one instance only, one can show that $\tau'_2$ approximates
$\tau_2$ well, this time in two variables.

\begin{lemma}
  Let $0\leq i$, $i\leq j$, and $i\leq k$. Then
  \begin{equation}
    \tau_2(v_i,q^{j-i},r^k)\leq\tau'_2(v_i,q^j,r^k)\leq\tau_2(v_i,q^j,r^{k-i}).
    \label{eq:2}
  \end{equation}
  \label{lemma:2}
\end{lemma}


\begin{proof}
We proceed as in the proof of Lemma \ref{lemma:1}. Note that the function
$\tau_2$ is monotone non-decreasing in $a$, but monotone non-increasing in
$L_1$. 
Proceeding by induction on $i$, the base case holds since $\tau_2(v_1,a,y) \leq
\tau_2'(v_1,b,y)$ for any $a \leq b$ and $y$. We suppose that Equation
(\ref{eq:2}) holds for all $p<i$. Then, for every $0 \leq \alpha < 1$,
\begin{eqnarray*}
  & & \tau_2'(p,q^{\lfloor j+\log_q\alpha\rfloor}, 
  r^{\lfloor\log_r(r^k-l)\rfloor})a \geq 
  \tau_2(p,q^{\lfloor j+\log_q\alpha\rfloor - p},
  r^{\lfloor\log_r(r^k-l)\rfloor}) \\
  & & \geq \tau_2(p, q^{j-p-1+\log_q\alpha},r^k-l) 
  \geq \tau_2(p, \alpha q^{j-i},r^k-l)\textrm{.}
\end{eqnarray*}
Thus, since every predecessor of $v_i$ has index smaller than $i$, 
\begin{eqnarray*}
  \tau'_2(v_i,q^j,r^k)&=&\min_{\substack{\alpha_1,\ldots,\alpha_d\\\sum\alpha_j=1}}
\max_s \tau'_2(p_s, q^{\lfloor j+\log_q\alpha_s\rfloor}, r^{\lfloor\log_r(r^k-l_s)\rfloor})+l'_s \\
  & \geq &
  \min_{\substack{\alpha_1,\ldots,\alpha_{d}\\\sum\alpha_j=1}}
\max_s \tau_2(p_s, \alpha_s q^{j-i},r^k-l_s)+l'_s
   =  \tau_2(v_i,q^{j-i},r^k)\textrm{.}
\end{eqnarray*}

The proof of the inequality $\tau_2'(v_i,q^j,r^k)\leq\tau_2(v_i,q^j,r^{k-i})$ is
similar. Assuming that (\ref{eq:2}) holds for every $p<i$, we obtain

\begin{eqnarray*}
  &&\tau_2'(p,q^{\lfloor j+\log_q\alpha\rfloor},r^{\lfloor\log_r(r^k-l)\rfloor})
  \leq \tau_2(p,q^{\lfloor j+\log_q\alpha\rfloor},r^{\lfloor\log_r(r^k-l)\rfloor-p}) \\
  &&\leq \tau_2(p,\alpha q^j,r^{\log_r(r^k-l)-p-1})
  \leq \tau_2(p,\alpha q^j,r^{k-i}-l)\textrm{.}
\end{eqnarray*}
Plugging it into the definition of $\tau'_2$, we obtain

\begin{eqnarray*}
  \tau_2'(v_i,q^j,r^k)&=&\min_{\substack{\alpha_1,\ldots,\alpha_d\\\sum\alpha_j=1}}
\max_s \tau_2'(p_s, q^{\lfloor j+\log_q\alpha_s\rfloor}, r^{\lfloor \log_r(r^k-l_s) \rfloor})+l'_s \\
  & \leq &
  \min_{\substack{\alpha_1,\ldots,\alpha_d\\\sum\alpha_j=1}}
\max_s \tau_2(p_s, \alpha_s q^{j}, r^{k-i}-l_s)+l'_s 
   =  \tau_2(v_i,q^{j}, r^{k-i})\textrm{.}
\end{eqnarray*}
\qed
\end{proof}

Using Lemma \ref{lemma:2}, we can show that $\tau_2'$ provides enough
information to compute an approximation of $\tau_2$. However, we cannot get a
$(1+\varepsilon)$ approximation to the optimal value as in Lemma~\ref{thm:1}, because we need to round the value of $L_1$ to a power of $r$ in
order for it to be legal parameter of $\tau_2'$ and we further round it during
the evaluation of $\tau_2'$. We will therefore relate the result of $\tau_2'$ to
the results of $\tau_2$ we would have gotten if we considered the value of $L_1$
when rounded up towards the nearest number that can be represented as $r^k$ for
integer $k$ and the value $r^{k-n}$.  Due to the choice of $r$, the ratio of
these two values is $1+\delta$.

\begin{lemma}
%
  Let $k$ be such that $\tau'_2(v_n,q^k, r^{\lceil\log_r L_1 \rceil})\leq L_2 <
  \tau'_2(v_n,q^{k+1}, r^{\lceil\log_r L_1\rceil})$, $a$ be such that
  $\tau_2(v_n,a,r^{\lceil\log_r L_1\rceil - n})\leq L_2 <
  \tau_2(v_n,a+1,r^{\lceil\log_r L_1\rceil - n})$, and $b$ be
  largest such that $\tau_2(v_n,b,r^{\lceil\log_r L_1\rceil})\leq L_2 <
  \tau_2(v_n,b+1,r^{\lceil\log_r L_1\rceil})$.
  Then $a\leq b$, $\frac{a}{q^k}\leq 1+\varepsilon$, and $\frac{q^k}{b} \leq
  1+\varepsilon$.
  \label{thm:2}
\end{lemma}

\begin{proof}
The statement that $a\leq b$ follows from the definition of $a$ and $b$:
decreasing the limit on the path length in the first instance from
$r^{\lceil\log_r L_1\rceil}$ to $r^{\lceil\log_r L_1\rceil - n}$ cannot increase
the number of possible paths. 
By applying Lemma \ref{lemma:2} twice, we get 
\begin{equation}
  \tau_2(v_n, q^{k-n}, r^{\lceil\log_r L_1\rceil}) \leq 
  \tau_2'(v_n,q^k,r^{\lceil\log_r L_1\rceil})\leq L_2\textrm{,}
  \label{eq:3}
\end{equation}
and
\begin{equation}
  L_2<\tau_2'(v_n,q^{k+1},r^{\lceil\log_r L_1\rceil}) \leq
  \tau_2(v_n,q^{k+1},r^{\lceil\log_r L_1\rceil - n})\textrm{.}
  \label{eq:4}
\end{equation}
From the definition of $a$ and (\ref{eq:4}) we can conclude $a\leq q^{k+1}$.
This implies that $\frac{a}{q^k}\leq q\leq 1+\varepsilon$, due to our choice of
$q$. 
Similarly, from the definition of $b$ and (\ref{eq:3}) we get $b\geq q^{k-n}$
and thus $\frac{q^k}{b}\leq q^n\leq 1+\varepsilon$.
\qed
\end{proof}

Lemma~\ref{thm:2} shows that the computed number of $s$-$t$ paths $q^k$ cannot
be larger than $b$ by more than a factor of $1+\varepsilon$, nor can it be
smaller than $a$ by a factor larger than $1+\varepsilon$. 
Furthemore, with the
aforementioned choice of $r$ as $\sqrt[n]{1+\delta}$, the difference between 
the rounded up value of $L_1$ which is $r^{\lceil\log_rL_1\rceil}$ and the 
rounded down value which is $r^{\lceil\log_rL_1\rceil-n}$ is $(1+\delta)$.
We can now state the
overall running time of the approach. Compared to the function $\tau'$ we need
to evaluate $\tau_2'$ for $\lceil\log_r L_1\rceil = O(n\delta^{-1}\log L_1)$
values of $r^l$, in addition to the values of $v_i$ and $q^k$. Otherwise the
arguments are identical to the proof of Theorem~\ref{thm:3}.  Note that $\log
L_1$ is by definition in $O(n)$, but we list it explicitly since it can be much
smaller in practice. 

\begin{lemma}
  Given path-lengths $L_1$ and $L_2$ for two given instances $G_1$ and $G_2$ of 
  a graph with $n$ edges and $m$ vertices, there is
  an algorithm that finds $k$ satisfying $\tau_2'(v_n,q^k,r^{\lceil\log_r
  L_1\rceil})\leq L < \tau_2'(v_n,q^{k+1},r^{\lceil\log_r L_1\rceil})$
  in time $O(mn^3\varepsilon^{-1}\delta^{-1}\log{n}\log{L_1})$.
  \label{thm:dag_bicriteria_runtime}
\end{lemma}

Putting together Lemma~\ref{thm:2} and Lemma~\ref{thm:dag_bicriteria_runtime} 
we can state the overall result: 

\begin{theorem}
For any $L_1$, $L_2$, any edge-weighted directed acyclic graphs on the same topology 
$G_1$ and $G_2$, and any two of their vertices $s$, $t$, there exists a length 
$L'_2$ satisfying $(1+\delta)^{-1}L_2\leq L'_2\leq L_2$ and an FPTAS
for counting the number of paths from $s$ to $t$ no longer 
than $L_1$ when evaluated on the graph $G_1$ and no longer than $L'_2$ when evaluated 
on the graph $G_2$ in the time $O(mn^4\varepsilon^{-1}\delta^{-1}\log{n}\log{L_1})$.
\end{theorem}

It is easy to see that we can extend the approach to count paths that
approximate $m$ instances at the same time by adding ``budgets''
$L_1,\ldots,L_{m-1}$ for the desired maximal lengths of paths in instances
$1,2,\ldots,m-1$. The time complexity would again increase, for every additional
instance with threshold $L_i$ by $O(n\delta^{-1}\log L_i)$.

\paragraph{Pseudo-polynomial algorithm for two instances.}

If the discrepancy between $a$ and $b$ as defined in Lemma~\ref{thm:2} is
too large and all edges have integer lengths, we can consider all possible
lengths in the first instance, instead of rounding to values in the form of
$r^k$. 


The function $\tau_2''$ will be $\tau'$ extended with the budget
representing the exact maximum length of a path in the first instance. 

\begin{eqnarray*}
&&\tau''_2(v_1,0,x)=-\infty, \forall x\in \mathbb{R^+} \\
&&\tau''_2(v_1,a,x)=0, \forall a: 0<a\leq 1, \forall x\in \mathbb{R^+} \\
&&\tau''_2(v_1,a,x)=\infty, \forall a: a > 1,\forall x\in \mathbb{R^+} \\
&&\tau''_2(v_i,q^j,r^k)=\min_{\substack{\alpha_1,\ldots,\alpha_d\\\sum\alpha_j=1}} 
\max_s \tau''_2(p_s,q^{\lfloor j+\log_q\alpha_s\rfloor}, L-l_s)+l'_s
\end{eqnarray*}

We will state the lemma and theorem about accuracy and runtime without proofs, since
these are similar to the proofs of Lemma~\ref{thm:1} and Theorem~\ref{thm:3}.
Notice that the algorithm evaluating $\tau''_2$ is pseudo-polynomial.

\begin{lemma}
Given $L$, let $k$ be such that 
$\tau_2''(v_n,q^k,L_1)\leq L_2 < \tau_2''(v_n,q^{k+1},L_1)$ and $a$ be such that 
$\tau_2(v_n,a,L_1)\leq L_2 < \tau_2(v_n,a+1,L_1)$. Then $(1+\varepsilon)^{-1}\leq \frac{a}{q^k} \leq 1+\varepsilon$.
\label{thm:5}
\end{lemma}

\begin{theorem}
  Given two graphs with integer weights, and any $\varepsilon > 0$, there is an
  algorithm that computes a $(1+\varepsilon)$-approximation for the number of
  $s$-$t$ paths that have length at most $L_1$ in the first instance, and length
  at most $L_2$ in the second instance, and runs in time
  $O(mn^3\varepsilon^{-1}L_1\log{n})$, where $m$ denotes the number of edges in
  the graph.
%
%
\label{thm:6}
\end{theorem}




\section{Concluding Remarks}
\label{sec:conclusions}

We have shown that there is an efficient algorithm to approximate the number of
approximately shortest paths in a directed acyclic graph.
This problem is implicitly or explicitly present as an algorithmic tool in
algorithmic solutions to a large number of different computational problems, not
limited to the evaluation of solutions achieved by dynamic programming which we
noted in Section \ref{sec:applications_of_shortest_paths_in_DAGs}.

Our result allows us, for instance, to approximately count only the small (or
large) terms of a polynomial $p(x)=\sum_i a_i x^i$, $a_i \geq 0$, represented as a
product $\prod_j p_j(x)$ of polynomially many polynomial factors $p_j(x)$, where
each $p_j(x)=\sum_k b_k x^k$ has polynomially many terms, and where every $b_k \geq
0$. This is especially interesting if the full expansion of $p(x)$ has
exponentially many terms.
This may be a powerful tool, if extended to the case of both negative and
positive $b_k$, enabling the counting of approximate solutions for problems with
known generating polynomials of solutions by weight. For instance, counting of
large graph matchings \cite{jerrum1987} or short spanning trees \cite{broder97}
can be done via generating polynomials (which, in general, have exponentially
many terms). This direction is our primary future work.  

We have also showed that our algorithm can be extended, given threshold weights
$L_1,\ldots, L_m$, and polynomially many graphs $G_1,\ldots,G_m$, to count
$s$-$t$ paths that have, at the same time, length at most $L_1$ in $G_1$ and 
at most $(1+\delta)L_i$ in $G_i$, $i=2,\ldots,m$. In the case when 
$m=2$, this algorithm is necessary for
application of the aforementioned robust optimization method \cite{buhmann2013}
to the various mentioned optimization problems.


\vspace{6pt}
\noindent{{\bf Acknowledgements.}
We thank Octavian Ganea and anonymous reviewers for their suggestions and
comments. The work has been partially supported by the Swiss National Science
Foundation under grant no.  200021\_138117/1, and by the EU FP7/2007-2013,
under the grant agreement no. 288094 (project eCOMPASS).}

\bibliographystyle{splncs03}  
\bibliography{main}

\begin{thebibliography}{10}
\providecommand{\url}[1]{\texttt{#1}}
\providecommand{\urlprefix}{URL }

\bibitem{broder97}
Broder, A.Z., Mayr, E.W.: Counting minimum weight spanning trees. J. Algorithms
   24,  171--176 (July 1997)

\bibitem{buhmann2013}
Buhmann, J.M., Mihal\'{a}k, M., \v{S}r\'amek, R., Widmayer, P.: Robust
  optimization in the presence of uncertainty. In: Proc. 4th Conference on
  Innovations in Theoretical Computer Sciencei (ITCS). pp. 505--514. ACM, New
  York, NY, USA (2013)

\bibitem{burge1997}
Burge, C., Karlin, S.: Prediction of complete gene structures in human genomic
  {DNA}. Journal of molecular biology  268(1),  78--94 (1997)

\bibitem{chen2001}
Chen, T., Kao, M.Y., Tepel, M., Rush, J., Church, G.M.: A dynamic programming
  approach to de novo peptide sequencing via tandem mass spectrometry. Journal
  of Computational Biology  8(3),  325--337 (2001)

\bibitem{durbin1998}
Durbin, R., Eddy, S.R., Krogh, A., Mitchison, G.: Biological sequence analysis:
  probabilistic models of proteins and nucleic acids. Cambridge university
  press (1998)

\bibitem{dyer1993mildly}
Dyer, M., Frieze, A., Kannan, R., Kapoor, A., Perkovic, L., Vazirani, U.: A
  mildly exponential time algorithm for approximating the number of solutions
  to a multidimensional knapsack problem. Combinatorics, Probability and
  Computing  2(3),  271--284 (1993)

\bibitem{gopalan2011}
Gopalan, P., Klivans, A., Meka, R., \v{S}tefankovi\v{c}, D., Vempala, S.,
  Vigoda, E.: An {FPTAS} for \# knapsack and related counting problems. In:
  Proc. 52nd Annual IEEE Symposium on Foundations of Computer Science (FOCS).
  pp. 817--826 (2011)

\bibitem{jerrum1987}
Jerrum, M.: Two-dimensional monomer-dimer systems are computationally
  intractable. Journal of Statistical Physics  48(1-2),  121--134 (1987)

\bibitem{karp1972}
Karp, R.M.: Reducibility among combinatorial problems. Springer (1972)

\bibitem{kreher98}
Kreher, D.L., Stinson, D.R.: Combinatorial Algorithms: Generation, Enumeration,
  and Search (1998)

\bibitem{lu03}
Lu, B., Chen, T.: A suboptimal algorithm for de novo peptide sequencing via
  tandem mass spectrometry. Journal of Computational Biology  10(1),  1--12
  (2003)

\bibitem{naor1993}
Naor, D., Brutlag, D.: On suboptimal alignments of biological sequences. In:
  Proc. 4th Annual Symposium on Combinatorial Pattern Matching (CPM). pp.
  179--196. Springer (1993)

\bibitem{stefankovic2012}
{\v{S}}tefankovi\v{c}, D., Vempala, S., Vigoda, E.: A deterministic
  polynomial-time approximation scheme for counting knapsack solutions. SIAM
  Journal on Computing  41(2),  356--366 (2012)

\bibitem{valiant79}
Valiant, L.G.: The complexity of computing the permanent. Theoretical computer
  science  8(2),  189--201 (1979)

\bibitem{Valiant/1979}
Valiant, L.G.: The complexity of enumeration and reliability problems. SIAM J.
  Comput.  8(3),  410--421 (1979)

\end{thebibliography}

\vfill\eject

\end{document}